\documentclass[a4paper,preprintnumbers,floatfix,superscriptaddress,pra,onecolumn,showpacs,notitlepage]{revtex4-1}

\usepackage[utf8]{inputenc}
\usepackage[T1]{fontenc}
\usepackage[sc,osf]{mathpazo}
\usepackage{amsmath, amsthm, amssymb,amsfonts}
\usepackage{graphicx}
\usepackage{dcolumn}
\usepackage{bm}
\usepackage{bbm}
\usepackage{hyperref}
\usepackage{comment}
\usepackage{color}
\usepackage{makecell}

\usepackage{microtype}
\usepackage{framed}

\def\1{\mathbf{1}}
\def\0{\mathbf{0}}

\DeclareMathOperator{\Tr}{Tr}

\newcommand{\ket}[1]{| #1 \rangle}
\newcommand{\bra}[1]{\langle #1 |}

\newcommand{\mean}[1]{\left\langle #1 \right\rangle}

\newtheorem{prop}{Proposition}
\newtheorem{lem}{Lemma}
\newtheorem{theo}{Theorem}

\newtheorem{theorem}[theo]{Theorem}

\newtheorem{lemma}[lem]{Lemma}

\newtheorem{proposition}[prop]{Proposition}

\renewcommand{\rho}{\varrho}

\newcommand{\processnext}[1]{%
  \ifx\listfinish#1\empty\else\listact{#1}\expandafter\processnext\fi}

\newcommand{\figref}[1]{Fig.~\ref{#1}}
\newcommand{\tabref}[1]{Table~\ref{#1}}

\renewcommand{\eqref}[1]{Eq. \ref{#1}}
\renewcommand{\figref}[1]{Fig.~\ref{#1}}
\renewcommand{\tabref}[1]{Table~\ref{#1}}

\newcommand{\theoref}[1]{Theorem~\ref{#1}}

\newcommand{\eq}[1]{
\begin{equation}
#1
\end{equation}
}

\newcommand{\arr}[1]{
\begin{array}{c}
#1
\end{array}
}

{\begin{framed}\begin{small}}
{\end{small}\end{framed}}

\DeclareGraphicsExtensions{.pdf,.png,.jpg}

\makeindex

\begin{document}
\title{
Maximal violation of n-locality inequalities in a star-shaped quantum network
}
\date{\today}

\author{Francesco Andreoli}
\affiliation{Dipartimento di Fisica - Sapienza Universit\`{a} di Roma, P.le Aldo Moro 5, I-00185 Roma, Italy}

\author{Gonzalo Carvacho}
\affiliation{Dipartimento di Fisica - Sapienza Universit\`{a} di Roma, P.le Aldo Moro 5, I-00185 Roma, Italy}

\author{Luca Santodonato}
\affiliation{Dipartimento di Fisica - Sapienza Universit\`{a} di Roma, P.le Aldo Moro 5, I-00185 Roma, Italy}

\author{Rafael Chaves}
\affiliation{International Institute of Physics, Federal University of Rio Grande do Norte, 59070-405 Natal, Brazil}

\author{Fabio Sciarrino}
\email{fabio.sciarrino@uniroma1.it}
\affiliation{Dipartimento di Fisica - Sapienza Universit\`{a} di Roma, P.le Aldo Moro 5, I-00185 Roma, Italy}

\begin{abstract}
Bell's theorem was a cornerstone for our understanding of quantum theory, and the establishment of Bell non-locality played a crucial role in the development of quantum information. Recently, its extension to complex networks has been attracting a growing attention, but a deep characterization of quantum behaviour is still missing for this novel context. In this work we analyze quantum correlations arising in the bilocality scenario, that is a tripartite quantum network where the correlations between the parties are mediated by two independent sources of states. First, we prove that non-bilocal correlations witnessed through a Bell-state measurement in the central node of the network form a subset of those obtainable by means of a separable measurement. This leads us to derive the maximal violation of the bilocality inequality that can be achieved by arbitrary two-qubit quantum states and arbitrary projective separable measurements. We then analyze in details the relation between the violation of the bilocality inequality and the CHSH inequality. Finally, we show how our method can be extended to $n$-locality scenario consisting of $n$ two-qubit quantum states distributed among $n+1$ nodes of a star-shaped network.
\end{abstract}

\maketitle

\section{Introduction} Since its establishment in the early decades of the last century, quantum theory has been elevated to the status of the ``most precisely tested and most successful theory in the history of science'' \cite{kleppner2000one}. And yet, many of its consequences have puzzled -- and still do-- most of the physicists confronted to it. At the heart of many of the counter-intuitive features of quantum mechanics is quantum entanglement \cite{horodecki2009quantum}, nowadays a crucial resource in quantum information and computation \cite{nielsen2002quantum} but that also plays a central role in the foundations of the theory. For instance, as shown by the celebrated Bell's theorem\cite{Bell1964}, quantum correlations between distant parts of an entangled system can violate Bell inequalities, thus precluding its explanation by any local hidden variable (LHV) model, the phenomenon known as quantum non-locality.

Given its fundamental importance and practical applications in the most varied tasks of quantum information \cite{Brunner2014}, not surprisingly many generalizations of Bell's theorem have been pursued over the years. Bell's original scenario involves two distant parties that upon receiving their shares of a joint physical system can measure one out of possible dichotomic observables. Natural generalizations of this simple scenario include more measurements per party \cite{Collins2004} and sequential measurements \cite{Gallego2014}, more measurement outcomes \cite{Collins2002}, more parties \cite{Mermin1990,Werner2001} and also stronger notions of quantum non-locality \cite{Svetlichny1987,Gallego2012,Bancal2013,Chaves2016causal}. All these different generalizations share the common feature that the correlations between the distant parties are assumed to be mediated by a single common source of states (see, for instance, Fig. \ref{fig_DAGS}a). However, as it is often in quantum networks \cite{Kimble2008}, the correlations between the distant nodes is not given by a single source but by many independent sources which distribute entanglement in a non-trivial way across the whole network and generate strong correlations among its nodes (Figs. \ref{fig_DAGS}b-d). Surprisingly, in spite of its clear relevance, such networked scenario is far less explored.

The simplest networked scenario is provided by entanglement swapping \cite{Zukowski1993}, where two distant parties, Alice and Charlie, share entangled states with a central node Bob (see \figref{fig_DAGS}b). Upon measuring in an entangled basis and conditioning on his outcomes, Bob can generate entanglement and non-local correlations among the two other distant parties even though they had no direct interactions. To contrast classical and quantum correlation in this scenario, it is natural to consider classical models consisting of two independent hidden variables (Figs. \ref{fig_DAGS}b), the so-called bilocality assumption \cite{Branciard2010,Branciard2012}. The bilocality scenario and generalizations to networks with an increasing number $n$ of independent sources of states (Figs. \ref{fig_DAGS}d), the so called n-locality scenario \cite{Tavakoli2014,mukherjee2015correlations,Chaves2016PRL,Rosset2016,tavakoli2016quantum,Tavakoli2016,paul2017detection,tavakoli2017correlations} allow for the emergence of a new kind of non-local correlations. For instance, correlations that appear classical according to usual LHV models can display non-classicality if the independence of the sources is taken into account, a result experimentally demonstrated in \cite{carvacho2016experimental,saunders2016experimental}. However, previous works on the topic have mostly focused on developing new tools for the derivation of inequalities characterizing such scenarios and much less attention has been given to understand what are the quantum correlations that can be achieved in such networks.

That is precisely the aim of the present work. We consider in details the bilocality scenario and the bilocality inequality derived in \cite{Branciard2010,Branciard2012} and characterize the non-bilocal behavior of general qubit quantum states when the parties perform different kinds of projective measurements. First of all we show that the correlations arising in an entanglement swapping scenario, i.e. when Bob performs a Bell-state measurement (BSM), form a strict subclass of those correlations which can be achieved by performing separable measurements in all stations. Focusing on this wider class of correlations, we derive a theorem characterizing the maximal violation of the bilocality inequality  \cite{Branciard2010,Branciard2012} that can be achieved from a general two-qubit quantum states shared among the parties. This leads us to obtain a characterization for the violation of the bilocality inequality in relation to the violation of the CHSH inequality \cite{Clauser1969}. Finally we show how our maximization method can be extended to the star network case \cite{Tavakoli2014}, a $n$-partite generalization of the bilocality scenario, deriving thus the maximum violation of the n-locality inequality that can be extracted from this network.

\section{Scenario} In the following we will mostly consider the bilocality scenario, which classical description in terms of directed acyclic graphs (DAGs) is shown in \figref{fig_DAGS}-b. It consists of three spatially separated parties (Alice, Bob and Charlie) whose correlations are mediated by two independent sources of states. In the quantum case, Bob shares two pairs of entangled particles, one with Alice and another with Charlie. Upon receiving their particles Alice, Bob and Charlie perform measurements labelled by the random variables $X$, $Y$ and $Z$ obtaining, respectively, the measurement outcomes $A$, $B$ and $C$. The difference between Bob and the other parties is the fact that the first has in his possession two particles and thus can perform a larger set of measurements including, in particular, measurements in an entangled basis.

\begin{figure}[t]
\centering
\includegraphics[width=\textwidth]{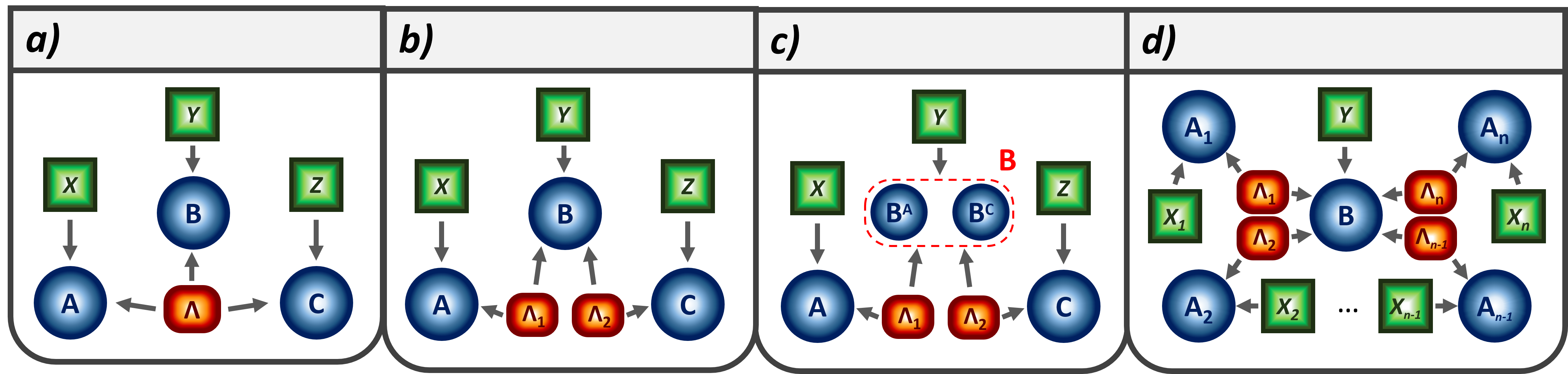}
\caption{\textbf{Description of the causal structure of some different networks.} \textbf{a)}  LHV model representing a tripartite scenario with a single source of states. \textbf{b)} BLHV model describing the bilocality counterpart of an entanglement swapping scenario. \textbf{c)} Causal structure of a bilocality scenario where the separable measurements performed in B are represented by the presence of the two substations $B^A$ and $B^C$. \textbf{d)} Extension of the bilocality scenario to a network consisting of \textit{n} different stations sharing a quantum state with a central node, i.e. the so-called n-local star network.}
\label{fig_DAGS}
\end{figure}

Any probability distribution compatible with the bilocality assumption (i.e. independence of the sources) can be decomposed as 
\eq{
\label{eq:bilocal_set_correlations_definition}
p(a,b,c|x,y,z)= \displaystyle \int d\lambda_1 d\lambda_2 p(\lambda_1) p(\lambda_2) p(a|x,\lambda_1)p(b|y, \lambda_1 ,\;\lambda_2)p(c|z, \lambda_2).}

In particular, if we consider that each party measures two possible dichotomic observables ($x,y,z,a,b,c=0,1$), it follows that any bilocal hidden variable (BLHV) model described by \eqref{eq:bilocal_set_correlations_definition} must fulfill the bilocality inequality
\begin{equation}
\label{eq:bilocality_inequality}
\mathcal{B}=\sqrt{|I|}+\sqrt{|J|} \leq 1,
\end{equation}
with
\eq{
\label{eq:IJ_definition}
\arr{I=\dfrac{1}{4}\displaystyle \sum_{x,z=0,1} \langle A_{x}B_{0}C_{z} \rangle,\;\;\;
 J=\dfrac{1}{4}  \displaystyle \sum_{x,z=0,1} (-1)^{x+z} \langle A_{x}B_{1}C_{z} \rangle,}}
and where
\eq{
\label{eq:mean_ABC_definition}
\langle A_{x}B_{y}C_{z} \rangle=   \displaystyle \sum_{a,b,c=0,1} (-1)^{a+b+c} p(a,b,c|x,y,z).}

As shown in \cite{Branciard2010,Branciard2012}, if we impose the same causal structure to quantum mechanics (e.g. in an entanglement swapping experiment) we can nonetheless violate the bilocality inequality (even though the data might be compatible with LHV models), thus showing the existence of a new form of quantum \textit{non-locality} called quantum \textit{non-bilocality}.

To that aim let us consider the entanglement swapping scenario with an overall quantum state $\ket{\psi^-}_{AB}\otimes\ket{\psi^-}_{BC}$, with $\ket{\psi^-}=(1/\sqrt{2})(\ket{01}-\ket{10})$. We can choose the measurements operators for the different parties in the following way. Stations A and C perform single qubit measurements defined by
\eq{
\label{eq:bilocality_swap_measurement_Branciard}
\arr{A_x=\dfrac{\sigma_z+(-1)^x\sigma_x}{\sqrt{2}},\;\;\;\;\;\;C_z=\dfrac{\sigma_z+(-1)^z\sigma_x}{\sqrt{2}}.
}
}
Station B, instead, performs a complete BSM, assigning to the two bits $b_0b_1$ the values

\eq{
\label{eq:bilocality_swap_measurement_Branciard_B}
\arr{
 00 \;\text{for} \;\ket{\phi^+}, \;\;\;\;\;01 \;\text{for}\; \ket{\phi^-},\;\;\;\;\;
 10 \;\text{for}\; \ket{\psi^+}, \;\;\;\;\;11 \;\text{for} \;\ket{\psi^-}.}
}

The binary measurement $B_y$ is then defined such that it returns $(-1)^{b_y}$, with respect to the value of $y=0,1$. This leads to
\eq{
\label{eq:mean_ABC_definition_swappCase}
\arr{
 \langle A_{x}B_{y}C_{z} \rangle= \displaystyle \sum_{a,b_{0},b_{1},c=0,1} (-1)^{a+b_{y}+c} p(a,b_{0},b_{1},c|x,z)\\\\=
\displaystyle \sum_{a,b_{y},c=0,1} (-1)^{a+b_{y}+c} p(a,b_{y},c|x,z) \equiv \displaystyle \sum_{a,b,c=0,1} (-1)^{a+b+c} p(a,b,c|x,y,z),}}
where, in the last steps, we made explicit use of the marginalization of probability $p(a,b_{0},b_{1},c|x,z)$ over $b_{k\neq y}$.\\
With these state and measurements, the quantum mechanical correlations achieve a value $\mathcal{B}=\sqrt{2}>1$, which violates the bilocality inequality and thus proves quantum non-bilocality.

\section{Results}
\subsection{Non-bilocal correlations with separable measurements}

As reproduced above, in an entanglement swapping scenario QM can exhibit correlations which cannot be reproduced by any BLHV model. In turn, it was recently proved \cite{Rosset2016} that an equivalent form of the bilocality inequality (\eqref{eq:bilocality_inequality}), can be violated by QM in the case where all parties only perform single qubit measurements (i.e. $\sigma_x, \;\sigma_z,\;\sigma_y $ and linear combinations). Here we will prove that, given the bilocality inequality (\eqref{eq:bilocality_inequality}), the non-bilocal correlations arising in an entanglement swapping scenario are a strict subclass of those obtainable by means of separable measurements.

The core of the bilocality parameter $\mathcal{B}$ is the evaluation of the expected value $\mean{A_xB_yC_z}$ (\eqref{eq:mean_ABC_definition}), that in the quantum case is given by 
	\eq{
		\label{eq:QM_mean_value}
		<A_xB_yC_z>=\Tr[(A_x\otimes B_y \otimes C_z)  ( \rho_{AB}\otimes\rho_{BC})].}
		
For the entanglement swapping scenario we can summarize the measurements in stations A and C by
	\eq{
		\label{eq:Ax_and_Cz_definitions}
		\arr{
		A_x=(1-x)\;A_0\; +\; x\;A_1\;\;\;\;\;x=0,1,\\\\
		C_z=(1-z)\;C_0\; +\; z\;C_1\;\;\;\;\;\;z=0,1,}
	}
where $A_x$ and $C_z$ are general single qubit projective measurements with eigenvalues $1$ and $-1$. When dealing with station B, it is suitable to consider its operatorial definition which is implicit in \eqref{eq:mean_ABC_definition_swappCase}. Indeed we can consider that $(-1)^{b_y}$ is the outcome of our measurement, leading to values shown in \tabref{tab:B_swap_values}.

\begin{table}[h!]
\label{tab:tomo}
\centering
\resizebox{0.65\textwidth}{!}{\begin{tabular}{| c | c | c | c | c |}
\hline
\diaghead{ooooooooooo}{$\hspace{0.1cm} y$}{$\vspace{0.2cm} b_0b_1$}  & 00($\phi^+$) & 01($\phi^-$) & 10($\psi^+$) & 11($\psi^-$)\\
\hline
$y=0$    &  1    &   1   &  -1    &  -1  \\
\hline
$y=1$    &  1    &  -1    &   1   &  -1  \\

\hline

\end{tabular}}
\vspace{7pt}
\caption{Expected values for the operator $B_y$, as implicitly defined in \eqref{eq:mean_ABC_definition_swappCase}.}
\label{tab:B_swap_values}
\end{table}

The quantum mechanical description of the operator $B_y$ (in an entanglement swapping scenario) is thus given by	
	\begin{eqnarray}
		\label{eq:By_swapping_definitions}
		B_y= \ket{\phi^+} \bra{\phi^+}+ (1-2y)\ket{\phi^-} \bra{\phi^-}
		+(2y-1)\ket{\psi^+} \bra{\psi^+} -\ket{\psi^-} \bra{\psi^-}
	\end{eqnarray}	
which relates each value of $y=0,\;1$ with its correct set of outcomes. This leads to the following theorem.

\begin{theorem}
[Non-bilocal correlations and separable measurements]
\label{theo:Bilo_sep=ent}
Given the general set of separable measurements
\eq{
\label{eq:B_y_separable_general_form}
B_y=(1-y)\sum_{ij}\lambda_{ij}\;\sigma_i\otimes \sigma_j+y\sum_{kl}\delta_{kl}\;\sigma_k\otimes \sigma_l,}
QM predictions for the bilocality parameter $\mathcal{B}$ which arise in an entanglement swapping scenario (where Bob performs the measurement described in \eqref{eq:By_swapping_definitions}) are completely equivalent to those obtainable by performing a strict subclass of \eqref{eq:B_y_separable_general_form}, i.e.
\eq{\{\mathcal{B}\}_{B.S.M.} \subset \{\mathcal{B}\}_{SEP.M.}.}
\end{theorem}

\begin{proof}
Let us write the Bell basis of a two qubit Hilbert space in terms of the computational basis ($\ket{00},\;\ket{01},\;\ket{10},\;\ket{11}$). From \eqref{eq:By_swapping_definitions}, we obtain

\eq{\arr{
		B_y= \ket{\phi^+} \bra{\phi^+}+ (1-2y)\ket{\phi^-} \bra{\phi^-}+(2y-1)\ket{\psi^+} \bra{\psi^+} -\ket{\psi^-} \bra{\psi^-}\\\\
	=	(1-y)\;(\ket{00} \bra{00}-\ket{01} \bra{01}-\ket{10} \bra{10}+\ket{11} \bra{11})+y\;(\ket{00} \bra{11}+\ket{01} \bra{10}+\ket{10} \bra{01}+\ket{11} \bra{00})\\\\=
(1-y)\;\sigma_z\otimes\sigma_z+y\;\sigma_x\otimes\sigma_x.		
}}

This shows that the entanglement swapping scenario is equivalent to the one where station $B$ only performs the two separable measurements $B_0=\sigma_z\otimes\sigma_z$ and $B_1=\sigma_x\otimes\sigma_x$, 
which form a strict subclass of the general set of separable measurements given by \eqref{eq:B_y_separable_general_form}. Moreover if we consider a rotated Bell basis, then we obtain
\eq{\arr{B_y'=U^{\dagger}_{AB}\otimes U^{\dagger}_{BC} B_y U_{AB}\otimes U_{BC}=
 (1-y)\;U^{\dagger}_{AB} \sigma_x U_{AB} \otimes U^{\dagger}_{BC}\sigma_x U_{BC}+y \;U^{\dagger}_{AB} \sigma_z U_{AB} \otimes U^{\dagger}_{BC}\sigma_z U_{BC}
 \\\\=(1-y)\;\vec{a}\cdot \vec{\sigma}\otimes\vec{c}\cdot \vec{\sigma}+y\; \vec{a}'\cdot \vec{\sigma} \otimes\vec{c}'\cdot \vec{\sigma}
 }}
 where $\vec{\sigma}=(\sigma_x,\sigma_y,\sigma_z)$ and $\vec{a},\;\vec{a}'$ ($\vec{c},\;\vec{c}'$) are orthogonal unitary vectors. Due to the constraints $\vec{a}\perp \vec{a}'$ and $\vec{c}\perp \vec{c}'$, this case still represents a strict subset of \eqref{eq:B_y_separable_general_form}. 
\end{proof}

As it turns out, this theorem has strong implications in our understanding of the non-bilocal behavior of QM. Indeed, it shows how the entanglement swapping scenario is not capable of exploring the whole set of quantum non-bilocal correlations, since it is totally equivalent to a subclass of Bob's separable measurements. As we will show next, a better characterization of quantum correlations within the bilocality context must thus in principle take into account the general form of Bob's separable measurements, especially when dealing with different types of quantum states.

\subsection{Non-bilocality maximization criterion} We will now explore the maximization of the bilocality inequality considering that Bob performs the separable measurements described by \eqref{eq:B_y_separable_general_form}. It is convenient to consider that station B as a unique station composed of the two substations $B^A$ and $B^C$, which perform single qubit measurements on one of the qubits belonging to the entangled state shared, respectively, with station A or C (see \figref{fig_DAGS}-c).\\
Let $A$ perform a general single qubit measurement and similarly for $B^A$, $B^C$ and $C$. We can define these measurements as
\eq{
\label{eq:general_separable_measurements}
\arr{
\text{Station A}\;\longrightarrow\;\;\vec{a}_x\cdot \vec{\sigma},\;\;\;\;
\text{Station B}\;\longrightarrow\;\;\vec{b}^A_y\cdot \vec{\sigma}\otimes\vec{b}^C_y\cdot \vec{\sigma},\;\;\;\;
\text{Station C}\;\longrightarrow\;\;\vec{c}_z\cdot \vec{\sigma},
}}
where $\vec{\sigma}=(\sigma_x,\sigma_y,\sigma_z)$. Let us now define a general 2-qubit quantum state density matrix as
\eq{\label{eq:general_2qubit_state}\rho=\dfrac{1}{4}(\mathbb{I}\otimes \mathbb{I}+\vec{r}\cdot \vec{\sigma}\otimes \mathbb{I}+\mathbb{I}\otimes\vec{s}\cdot \vec{\sigma}+\sum^3_{n,\;m=1}t_{nm}\;\sigma_n \otimes \sigma_m).}
The coefficients $t_{nm}$ can be used to define a real matrix $T_{\rho}$ that lead to the following result:
\begin{lemma}[Bilocality Parameter with Separable Measurements]
\label{theo:B_separable_measurements}
\noindent
Given the set of general separable measurements described in \eqref{eq:general_separable_measurements} and defined the general quantum state $\rho_{AB} \otimes \rho_{BC}$ accordingly to \eqref{eq:general_2qubit_state}, the bilocality parameter $\mathcal{B}$ is given by 
\eq{
\label{eq:B_separable_measurements}
\begin{array}{c}
\mathcal{B}=\dfrac{1}{2}\sqrt{\big| (\vec{a}_0+\vec{a}_1)\cdot T_{\rho_{AB}}\vec{b}^A_0 \big| \;\;\big| \vec{b}^C_0\cdot T_{\rho_{BC}}(\vec{c}_0+\vec{c}_1) \big|}+
\dfrac{1}{2}\sqrt{\big| (\vec{a}_0-\vec{a}_1)\cdot T_{\rho_{AB}}\vec{b}^A_1 \big| \;\;\big| \vec{b}^C_1\cdot T_{\rho_{BC}}(\vec{c}_0-\vec{c}_1) \big|}.
\end{array}
}
\end{lemma}

\begin{proof}
Let us consider two operators $O_i$ in the form $O_i=\vec{v}_i\cdot \vec{\sigma}$ and a two qubit quantum state $\rho$ described by \eqref{eq:general_2qubit_state}. We can write
\eq{
\label{eq:operators_and_vectors_proof}
\arr{\mean{O_1\otimes O_2}_{\rho}=\Tr[(O_1\otimes O_2)\rho]=\Tr[\displaystyle \sum_{j,k=1,2,3}(v_1^j v_2^k \sigma_j \otimes \sigma_k)\rho]=
\displaystyle \sum_{j,k=1,2,3} v_1^j v_2^k t_{jk}=\vec{v}_1\cdot(T_{\rho}\vec{v}_2),
}
}
where we made use of the properties of the Pauli matrices $\sigma_i$. Given the set of separable measurements described in \eqref{eq:general_separable_measurements}, and the definitions of $I$ and $J$ (showed in \eqref{eq:IJ_definition}), the proof comes from a direct application of \eqref{eq:operators_and_vectors_proof} to the quantum mechanical expectation value:
\eq{\mean{A_x \otimes B^A_y \otimes B^C_y \otimes C_z}_{\rho_{AB}\otimes \rho_{BC}}= \mean{A_x \otimes B^A_y}_{\rho_{AB}}\mean{ B^C_y \otimes C_z }_{\rho_{BC}} .}

\end{proof}

Next we proceed with the maximization of the parameter $\mathcal{B}$ over all possible measurement choices, that is, the maximum violation of bilocality we can achieve with a given set of quantum states.
To that aim, we introduce the following Lemma.

\begin{lemma}
\label{lemma:MMT_MTM}
Given a square matrix $M$ and defined the two symmetric matrices $\mathcal{M}_1=M^\mathbf{T}M$ and $\mathcal{M}_2=MM^\mathbf{T}$, each \textit{non-null} eigenvalue of $\mathcal{M}_1$ is also an eigenvalue of $\mathcal{M}_2$, and \textit{vice versa}.
\end{lemma}

\begin{proof}
Let $\lambda$ be an eigenvalue of $\mathcal{M}_1$
\eq{M^\mathbf{T}M \vec{v}=\lambda\vec{v}.}
If $\lambda \neq 0$ we must have $M\vec{v}\neq \vec{0}$. We can then apply the operator $M$ from the left, obtaining
\eq{MM^\mathbf{T} (M\vec{v})=\lambda(M\vec{v}),}
which shows that $M\vec{v}$ is an eigenvector of $\mathcal{M}_2$ with eigenvalue $\lambda$. 
\\The \textit{opposite} statement can be analogously proved.
\end{proof}

We can now enunciate the main result of this section.

\begin{theorem}[Bilocality Parameter Maximization]
\label{theo:B_maximization_separable}
Given the set of general separable measurements described in \eqref{eq:general_separable_measurements}, the maximum bilocality parameter that can be extracted from a quantum state $\rho_{AB} \otimes \rho_{BC}$ can be written as
\eq{\arr{\label{eq:max_bilocality}
\mathcal{B}_{max}=\sqrt{\sqrt{t^A_1 t^C_1} + \sqrt{t^A_2 t^C_2}},}
}
where $t^A_1$ and $t^A_2$ ($t^C_1$ and $t^C_2$) are the two greater (and positive) eigenvalues of the matrix $T_{\rho_{AB}}^\mathbf{T} T_{\rho_{AB}}$ ($T_{\rho_{BC}}^\mathbf{T} T_{\rho_{BC}}$), with $t^A_1\geq t^A_2$ and $t^C_1\geq t^C_2$.
\end{theorem}

\begin{proof}
We will prove \theoref{theo:B_maximization_separable}, following a scheme similar to the one used by Horodecki \cite{Horodecki1995} for the CHSH inequality. Let us introduce the two pairs of mutually orthogonal vectors
\eq{
\label{eq:orthogonal_vectors_changes}\arr{
(\vec{a}_0+\vec{a}_1)=2\; \cos \alpha\; \vec{n}_A \;\;\;\&\;\;\; (\vec{a}_0-\vec{a}_1)=2\; \sin \alpha\; \vec{n}'_A,\\\\
(\vec{c}_0+\vec{c}_1)=2\; \cos \gamma\; \vec{n}_C \;\;\;\&\;\;\; (\vec{c}_0-\vec{c}_1)=2\; \sin \gamma\; \vec{n}'_C,
}}

and let us apply \eqref{eq:orthogonal_vectors_changes} to \eqref{eq:B_separable_measurements}

\eq{
\label{eq:B_max_proof}
\arr{

\mathcal{B}_{max}=\max \big(\sqrt{\big|  \big( \vec{n}_A \cdot T_{\rho_{AB}} \vec{b}^A_0 \big)\;  \big(\vec{b}^C_0\cdot T_{\rho_{BC}} \vec{n}_C \big)\;\cos\alpha\;\cos\gamma  \big|}+
\sqrt{\big|\big( \vec{n}'_A \cdot T_{\rho_{AB}} \vec{b}^A_1 \big)\;  \big( \vec{b}^C_1\cdot T_{\rho_{BC}} \vec{n}'_C \big)\; \sin\alpha\;\sin\gamma\big|}\big)}}

\eq{\arr{
\nonumber
=\max \big(\sqrt{\big|  \big( \vec{b}^A_0\cdot T^{\mathbf{T}}_{\rho_{AB}}\vec{n}_A \big)\;  \big(\vec{b}^C_0\cdot T_{\rho_{BC}} \vec{n}_C \big)\;\cos\alpha\;\cos\gamma  \big|}+
\sqrt{\big|\big(\vec{b}^A_1 \cdot T^{\mathbf{T}}_{\rho_{AB}}\vec{n}'_A  \big)\;  \big( \vec{b}^C_1\cdot T_{\rho_{BC}} \vec{n}'_C \big)\; \sin\alpha\;\sin\gamma\big|}\big),
}}
where the maximization is done over the variables $\vec{n}_A,\;\vec{n}'_A,\;\vec{b}^A_0,\;\vec{b}^A_1,\;
\vec{n}_C,\;\vec{n}'_C,\;\vec{b}^C_0,\;\vec{b}^C_1,\;\alpha$ and $\gamma$. We can choose $\vec{b}^A_0,\;\vec{b}^A_1,\;\vec{b}^C_0,$ and $\vec{b}^C_1 $ so that they maximize the scalar product. Defining 
\eq{|| M\vec{v}||^2=M\vec{v}\cdot M\vec{v}=\vec{v}\cdot M^{\mathbf{T}}M\vec{v},}
and remembering that $\vec{b}^A_0,\;\vec{b}^A_1,\;\vec{b}^C_0,$ and $\vec{b}^C_1 $ are unitary vectors, we obtain

\eq{\arr{
\mathcal{B}_{max}=\max\big(
\sqrt{|| T^{\mathbf{T}}_{\rho_{AB}}\vec{n}_A ||\; || T_{\rho_{BC}} \vec{n}_C||\;\big| \cos\alpha\;\cos\gamma  \big|}+
\sqrt{ || T^{\mathbf{T}}_{\rho_{AB}}\vec{n}'_A ||\;  || T_{\rho_{BC}} \vec{n}'_C ||\; \big| \sin\alpha\;\sin\gamma\big|}\big).
}}

Next we have to choose the optimum variables variables $\alpha$ and $\gamma$. This leads to the set of equations
\eq{
\label{eq:colored_white_imperfect_e_partial_derivatives_2values}
\begin{array}{l}
\dfrac{\partial \mathcal{B}(\alpha,\gamma)}{\partial \alpha}=\dfrac{1}{2}\sqrt{|| T^{\mathbf{T}}_{\rho_{AB}}\vec{n}'_A ||\; || T_{\rho_{BC}} \vec{n}'_C||\; \big|\sin(\alpha)\sin(\gamma)\big|}\;\cot(\alpha)-\dfrac{1}{2}\sqrt{|| T^{\mathbf{T}}_{\rho_{AB}}\vec{n}'_A ||\; || T_{\rho_{BC}} \vec{n}'_C||\; \big|\cos(\alpha)\cos(\gamma)\big|}\;\tan(\alpha)=0,\\\\
\dfrac{\partial \mathcal{B}(\alpha,\gamma)}{\partial \gamma}=\dfrac{1}{2}\sqrt{|| T^{\mathbf{T}}_{\rho_{AB}}\vec{n}'_A ||\; || T_{\rho_{BC}} \vec{n}'_C||\; \big|\sin(\alpha)\sin(\gamma)\big|}\;\cot(\gamma)-\dfrac{1}{2}\sqrt{|| T^{\mathbf{T}}_{\rho_{AB}}\vec{n}'_A ||\; || T_{\rho_{BC}} \vec{n}'_C||\; \big|\cos(\alpha)\cos(\gamma)\big|}\;\tan(\gamma)=0.
\end{array}
}
This system of equations admits only solutions constrained by
\eq{
\label{eq:colored_white_imperfect_e_partial_derivatives_2values_sol}
\tan(\alpha)^2=\tan(\gamma)^2\;\;\leftrightarrow\;\;\gamma = \pm \alpha + n\pi \;,\;n\in\mathbb{Z},
}
leading to
\eq{\label{eq:colored_imperfect_e_simplify}
\begin{array}{c}
\mathcal{B}_{max}=\max\big( |\cos \alpha|\sqrt{|| T^{\mathbf{T}}_{\rho_{AB}}\vec{n}_A ||\; || T_{\rho_{BC}} \vec{n}_C|| }+| \sin\alpha|\sqrt{|| T^{\mathbf{T}}_{\rho_{AB}}\vec{n}'_A ||\; || T_{\rho_{BC}} \vec{n}'_C|| }\big)
\\\\=\max\big(

\sqrt{|| T^{\mathbf{T}}_{\rho_{AB}}\vec{n}_A ||\; || T_{\rho_{BC}} \vec{n}_C|| + || T^{\mathbf{T}}_{\rho_{AB}}\vec{n}'_A ||\;  || T_{\rho_{BC}} \vec{n}'_C ||}\big)
\end{array}}

Next, we must take into account the constraints $\vec{n}_A \perp \vec{n}'_A$ and $\vec{n}_C \perp \vec{n}'_C$. Since these two couples of vectors are, however, independent, we can proceed with a first maximization which deals only with the two set of variables $\vec{n}_A$ and $\vec{n}'_A$. Since $T_{\rho_{AB}}T^{\mathbf{T}}_{\rho_{AB}}$ is a symmetric matrix, it is diagonalizable. Let us call $\lambda_1,\;\lambda_2$ and $\lambda_3$ its eigenvalues and let us write $\vec{n}_A$ and $\vec{n}'_A$ in an eigenvector basis. If we define $k_1=|| T_{\rho_{BC}} \vec{n}_C||>0$ and $k_2=|| T_{\rho_{BC}} \vec{n}'_C ||>0$, our problem can be written in terms of Lagrange multipliers related to the maximization of a function $f$, given the constraints $g_i$

\eq{
\label{eq:lagrange_problem_bilo}
\arr{
f(\vec{n}_A,\;\vec{n}'_A)=k_1\;\sqrt{ \displaystyle \sum_{i=1,2,3} \lambda_i (n^i_A)^2 }\; + \; k_2\; \sqrt{\displaystyle \sum_{i=1,2,3} \lambda_i (n'^i_A)^2},\\\\
g_1(\vec{n}_A)=\vec{n}_A\cdot\vec{n}_A-1,\;\;\;\;\;
g_2(\vec{n}'_A)=\vec{n}'_A\cdot\vec{n}'_A-1,\;\;\;\;\;
g_3(\vec{n}_A,\;\vec{n}'_A)=\vec{n}_A\cdot\vec{n}'_A,}}
where we considered that finding the values that maximize $\sqrt{|f(x)|}$ is equivalent to find these values for $|f(x)|$. Let us now introduce the scaled vectors $\vec{\eta}_A=k_1\;\vec{n}_A$ and $\vec{\eta}\;'_A=k_2\;\vec{n}\;'_A$. We obtain

\eq{\arr{
f(\vec{\eta}_A,\;\vec{\eta}\;'_A)= \sqrt{\displaystyle \sum_{i=1,2,3} \lambda_i (\eta^i_A)^2} +\sqrt{\displaystyle \sum_{i=1,2,3} \lambda_i (\eta\;'^i_A)^2},\\\\
g_1(\vec{\eta}_A)=\vec{\eta}_A\cdot\vec{\eta}_A-(k_1)^2,\;\;\;\;\;
g_2(\vec{\eta}\;'_A)=\vec{\eta}\;'_A\cdot\vec{\eta}\;'_A-(k_2)^2,\;\;\;\;\;
g_3(\vec{\eta}_A,\;\vec{\eta}\;'_A)=\vec{\eta}_A\cdot\vec{\eta}\;'_A,}}
whose solution is given by vectors with two null components, out of three. If we define $\lambda_1\geq\lambda_2\geq\lambda_3$ and if $k_1>k_2$, the solution related to the maximal value is then given by 
\eq{
f_{max}=k_1 \sqrt{\lambda_1} \; + \; k_2\; \sqrt{\lambda_2}
}
which leads to
\eq{
\mathcal{B}_{max}=\max_{\vec{n}_C,\;\vec{n}'_C}\big(
\sqrt{|| T_{\rho_{BC}} \vec{n}_C|| \;\sqrt{t^A_1}\;+ \;  || T_{\rho_{BC}} \vec{n}'_C ||\;\sqrt{t^A_2}}\big),
}
where we made use of the Lemma \ref{lemma:MMT_MTM}.\\
The maximization over the last two variables leads to an analogous Lagrange multipliers problem with similar solutions, thus proving the theorem.
\end{proof}
This theorem generalizes the results of \cite{Mukherjee2016} (which dealt with some particular classes of quantum states in the entanglement swapping scenario) to the more generic case of any quantum state in the separable measurements scenario (which, in a bilocality context, includes the correlations obtained through entanglement swapping). It represents an extension of the Horodecki criterion \cite{Horodecki1995} to the bilocality scenario, taking into account the general class of separable measurements which can be performed in station B. Our result thus shows that as far as we are concerned with the optimal violations of the bilocality inequality provided by given quantum states, separable measurements or a BSM (in the right basis) are fully equivalent.

\subsection{The relation between the non-bilocality and non-locality of sources}
We will now characterize quantum non-bilocal behaviour with respect to the usual non-locality of the states shared between A, B and B, C. Let us start from \eqref{eq:max_bilocality} and separately consider Bell non-locality of the states $\rho_{AB}$ and $\rho_{BC}$. We can quantify it by evaluating the greatest CHSH inequality violation that can be obtained with these states. Let us define the CHSH inequality as
\eq{
\label{eq:CHSH}
\begin{array}{l}
\mathcal{S}^{UV} \equiv \dfrac{1}{2}|\mean{U_0V_0+U_0V_1+U_1V_0-U_1V_1}|\leq 1.
\end{array}}

If we apply the criterion by Horodecki \textit{et al.} \cite{Horodecki1995}, we obtain

\eq{
\label{eq:max_CHSH}
\begin{array}{l}
\mathcal{S}^{AB}_{max}=\sqrt{t^A_1+t^A_2},\;\;\;\;\;
\mathcal{S}^{BC}_{max}=\sqrt{t^C_1+t^C_2},
\end{array}}

where we defined $t^A_1,\;t^A_2,\;t^C_1$ and $t^C_2$ accordingly to \eqref{eq:max_bilocality}. From a direct comparison of \ref{eq:max_bilocality} and \ref{eq:max_CHSH} we can write

\begin{proposition}
\label{prop:loc_means_biloc}
\eq{\mathcal{S}^{AB}_{max}\leq 1 \;\;\& \;\;\mathcal{S}^{BC}_{max}\leq 1 \;\;\longrightarrow \;\;\mathcal{B}_{max}\leq 1.} 
\end{proposition}

\begin{proof} Applying the Cauchy-Schwarz inequality we obtain 
\eq{\mathcal{B}^2_{max}\leq \mathcal{S}^{AB}_{max} \mathcal{S}^{BC}_{max}\leq 1.}
\end{proof}

\begin{figure}
\centering
\includegraphics[width=\textwidth]{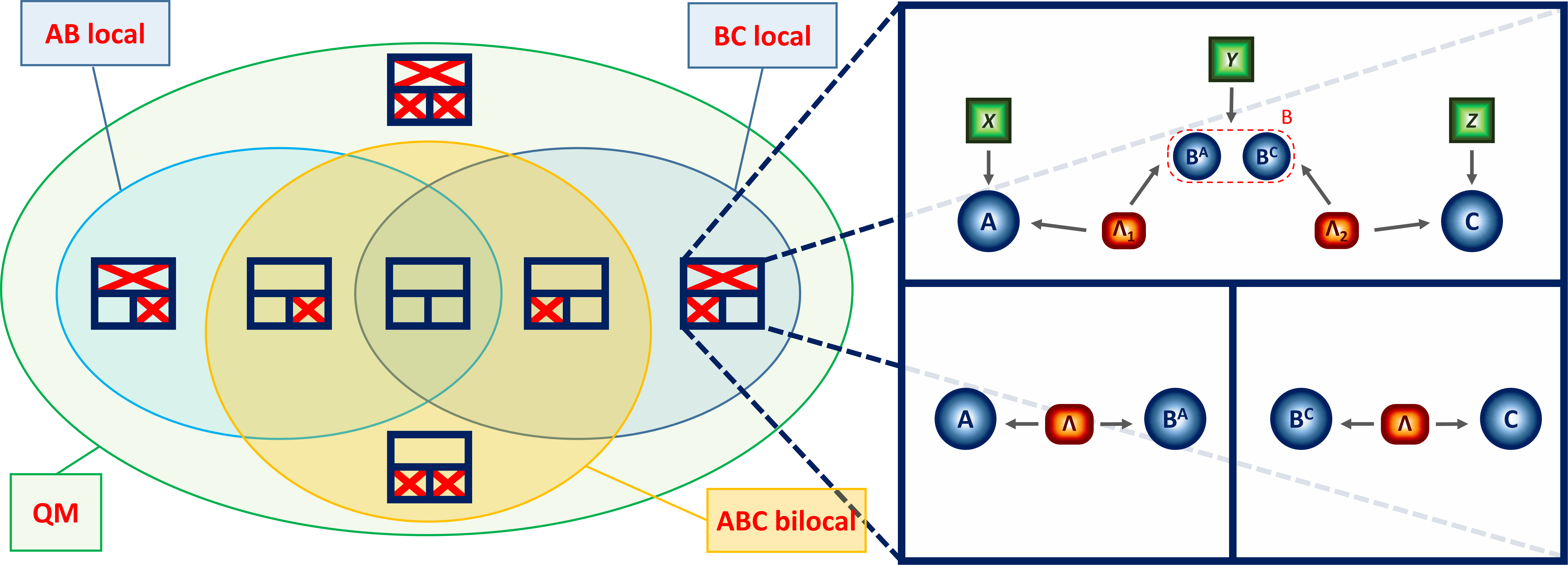}
\caption{\textbf{Venn’s diagram representing quantum correlations in a bilocality scenario.} Possible quantum correlations that may be witnessed given a quantum state $\rho_{AB}\otimes\rho_{BC}$. The blue sets represent quantum states that do not violation the CHSH inequality for $\rho_{AB}$ (AB local) or $\rho_{BC}$ (BC local). The orange set includes, instead, these states whose correlations do not violate the bilocality inequality, while the whole set of quantum correlations is represented in green. For all different regions a blue square shows those decompositions which are not allowed (crossed with red lines), accordingly to the greater square on the right.}
\label{fig_Exp2_New_Loc-Biloc}
\end{figure}

This result shows that if the two sources cannot violate the CHSH inequality then they will also not violate the bilocality inequality. Thus, in this sense, if our interest is to check the non-classical behaviour of sources of states, it is just enough to check for CHSH violations (at least if Bob performs a BSM or separable measurements). Notwithstanding, we highlight that this does not mean that the bilocality inequality is useless, since there are probability distributions that violate the bilocality inequality but nonetheless are local according to a LHV model and thus cannot violate any usual Bell inequality.

Next we consider the reverse case: is it possible to have quantum states that can violate the CHSH inequality but cannot violate the bilocality inequality? That turns out to be the case. To illustrate this phenomenon, we start considering two Werner states in the form $\rho=v(\ket{\psi^-}\bra{\psi^-})+(1-v)\mathbb{I}/4$. In this case, indeed, in order to have a non-local behaviour between A and B (B and C) we must have $v_{AB}> 1/\sqrt{2}$ ($v_{BC}> 1/\sqrt{2}$) while it is sufficient to have $\sqrt{v_{AB} v_{BC}}> 1/\sqrt{2}$ in order to witness non-bilocality. This example shows that on one hand it might be impossible to violate the bilocality inequality although one of $\rho_{AB}$ or $\rho_{BC}$ is Bell non-local (for instance $v_A=1$ and $v_C=0$). It also shows that, when one witnesses non-locality for only one of the two states, it can be possible, at the same time, to have non-bilocality by considering the entire network (for instance $v_A=1$ and $1/2<v_C<1/\sqrt{2}$).

Another possibility is the one described by the following Proposition
\begin{proposition}
Given a tripartite scenario
\eq{\exists\;\rho_{AB}\;\&\;\rho_{BC}\;\;\text{such that}\;\;\;\mathcal{S}^{AB}_{max}> 1,\; \;\mathcal{S}^{BC}_{max}> 1\; \;\&\;\;\mathcal{B}_{max}\leq 1.}
\end{proposition}

\begin{proof}
We will prove this point with an example. Let us take

\eq{
\arr{
\rho_{AB}=\dfrac{3}{5} \ket{\psi^+}\bra{\psi^+}+\dfrac{2}{5} \ket{\phi^+}\bra{\phi^+}=\left( \begin{array}{cccc} 0.2& 0 & 0 & 0.2 \\ 0 & 0.3 & 0.3 & 0 \\ 0 & 0.3 & 0.3 & 0 \\0.2& 0 & 0 & 0.2 \end{array}\right),\\\\
\rho_{BC}= \rho(v=\dfrac{7}{10},\;\lambda =\dfrac{1}{3})=\left( \begin{array}{cccc} 0.5& 0 & 0 & 0 \\ 0 & 0.45 & -0.35 & 0 \\ 0 & -0.35 & 0.45 & 0 \\0& 0 & 0 & 0.5 \end{array}\right),
}
}
where we defined $\rho(v,\lambda)$ as
\eq{
\label{eq:colored_white_noise_states}
\begin{array}{l}
\rho(v,\lambda) = \;v\; \ket{\psi^-}\bra{\psi^-}\;+\;(1-v)[\lambda\dfrac{\ket{\psi^-}\bra{\psi^-}+\ket{\psi^+}\bra{\psi^+}}{2}+(1-\lambda)\dfrac{\mathbb{I}}{4}].
\end{array}}
For these two quantum states one can check that
\eq{t_1^A=1,\;\;\;\;t_2^A=0.04,\;\;\;\;t_1^C=0.64,\;\;\;\;t_2^C=0.49,}
which leads to
\eq{\mathcal{S}^{AB}_{max}\simeq 1.02,\; \;\mathcal{S}^{BC}_{max}\simeq 1.06,\; \;\mathcal{B}_{max}\simeq 0.97.}
\end{proof}

This shows how it is possible to have non-local quantum states which nonetheless cannot violate the bilocality inequality (with separable measurements).

All these statements provide a well-defined picture of the relation between the CHSH inequality and the bilocality inequality in respect to the quantum states $\rho_{AB}\otimes\rho_{BC}$. We indeed derived all the possible cases of quantum non-local correlations which may be seen between couples of nodes, or in the whole network (according to the CHSH and bilocality inequalities). This characterization is shown in \figref{fig_Exp2_New_Loc-Biloc}, in terms of a Venn diagram.

We finally notice that if A and B share a maximally entangled state while B and C share a generic quantum state, then it is easier to obtain a bilocality violation in the tripartite network rather than a CHSH violation between the nodes $B^C$ and $C$. Indeed it is possible to derive
\eq{\begin{array}{l} \mathcal{B}_{max}(\ket{\Phi^{+}}\bra{\Phi^{+}}\otimes\rho_{BC})=\sqrt{\sqrt{t^C_1}+\sqrt{t^C_2}}\geq \sqrt{t^C_1+t^C_2}=\mathcal{S}^{BC}_{max}\end{array},}
where we made use of the following Lemma
\begin{lemma}
\label{lemma:t_less_1}
Given the parameters $t^A_1,\;t^A_2,\;t^C_1$ and $t^C_2$ defined in \eqref{eq:max_bilocality}, it holds
\eq{0\leq t^A_1,\;t^A_2,\;t^C_1,\;t^C_2 \leq 1}
\end{lemma}

\begin{proof}
This proof will be divided in two main points.\\\\
\textbf{1) } $\forall \rho,\;\;\exists \;\rho'=U^\dagger \rho U\;\text{such that}\;T_{\rho'}\;is\;diagonal$.\\
\textit{As discussed in \cite{Makhlin2002}, if we apply a local unitary $U=U_1\otimes U_2$ to the initial quantum state $\rho$, the matrix $T_{\rho}$ will transform accordingly to
\eq{T_{\rho}\longrightarrow U_1 T_{\rho} U_2^{\mathbf{T}}.}
According to the Singular Decomposition Theorem, it is always possible to choose $U_1$ and $U_2$ such that $U_1 T_{\rho} U_2^{\mathbf{T}}$ is diagonal, thus demonstrating point 1.}
\\\\
It is important to stress that we can always rotate our Hilbert space in a way that $\rho \rightarrow U^\dagger \rho U$ so we can take $\rho'$ without loss of generality.\\\\
\textbf{2) } \textit{If $T_{\rho}$ is diagonal, then the eigenvalues of $T_{\rho}^{\mathbf{T}}T_{\rho}$ are less or equal to 1.\\
It was shown in \cite{Horodecki1995} that, for every quantum state $\rho$, we have $|t_{nm}| \leq 1, \;t_{nm}\in \mathcal{R}$ regardless to the basis chosen for our Hilbert space. If $T_{\rho}$ is diagonal then $T_{\rho}^{\mathbf{T}}T_{\rho}=T_{\rho}^2$ and its eigenvalues $t_i$ can be written as $t_i= t_{ii}^2 \leq 1$.}\\\\
Given the definitions of $t^A_1$ and $t^A_2$ ($t^C_1$ and $t^C_2$) described in \eqref{eq:max_bilocality}, the lemma is proved.  
\end{proof}

\subsection{Extension to the star network scenario}
We now generalize the results of \theoref{theo:B_maximization_separable}, to the case of a \textit{n-partite} star network. This network is the natural extension of the bilocality scenario, and it is composed of $n$ sources sharing a quantum state between one of the $n$ stations $A_i$ and a central node B (see \figref{fig_DAGS}-d). The bilocality scenario corresponds to the particular case where $n=2$. The classical description of correlations in this scenario is characterized by the probability decomposition
\eq{
\label{eq:star_network_decomposition}
p(\{a_i\}_{i=1,n},b|\{x_i\}_{i=1,n}, y) = \displaystyle \int \bigg( \displaystyle \prod_{i=1}^n d\lambda_i p(\lambda_i) p(a_i|x_i,\lambda_i) \bigg) p(b|y,\{\lambda_i\}_{i=1,n}).}

As shown in \cite{Tavakoli2014}, assuming binary inputs and outputs in all the stations, the following n-locality inequality holds
\eq{\mathcal{N}_{star}=|I|^{1/n}+|J|^{1/n}\leq 1,}
where
\eq{\arr{
I=\dfrac{1}{2^n}\displaystyle \sum_{x_1...x_n}\mean{A^1_{x_1}...A^n_{x_n}B_0},\;\;\;\;\;
I=\dfrac{1}{2^n}\displaystyle \sum_{x_1...x_n}(-1)^{\sum_i x_i}\mean{A^1_{x_1}...\;A^n_{x_n}B_1},\\\\
\mean{A^1_{x_1}...A^n_{x_n}B_y}=\displaystyle \sum_{a_1...a_n,b} (-1)^{b+\sum_i a_i} p(\{a_i\}_{i=1,n}, b|\{x_i\}_{i=1,n}, y).}}
We will now derive a theorem showing the maximal value of parameter $\mathcal{N}_{star}$ that can be obtained by separable measurements on the central node and given arbitrary bipartite states shared between the central node and the $n$ parties.

\begin{theorem}[Optimal violation of the n-locality inequality]
\label{theo:star_netw_maximization}
Given single qubit projective measurements and defined the generic quantum state $\rho_{A_1B} \otimes...\otimes \rho_{A_{n}B}$ accordingly to \eqref{eq:general_2qubit_state}, the maximal value of $\mathcal{N}_{star}$ is given by
\eq{\arr{\label{eq:max_star_network}
\mathcal{N}_{star}^{max}=
\displaystyle \sqrt{(\displaystyle \prod_{i=1}^n t^{A_i}_1)^{1/n} + (\displaystyle \prod_{i=1}^n t^{A_i}_2)^{1/n}},
}}
where $t^{A^i}_1$ and $t^{A^i}_2$ are the two greater (and positive) eigenvalues of the matrix $T_{\rho_{A_i B}}^\mathbf{T} T_{\rho_{A_i B}}$ with $t^{A_i}_1 \geq t^{A_i}_2$.
\end{theorem}

\begin{proof}
In our single qubit measurements scheme the operator $B$ can be written as
\eq{B_y=\displaystyle \bigotimes_{i=1}^{n}B^i_y=\displaystyle \bigotimes_{i=1}^{n} \vec{b}^i_y\cdot \vec{\sigma}.}
As pointed out in \cite{Tavakoli2014}, this allows us to write
\eq{\arr{
\mathcal{N}_{star}=|\displaystyle \prod_{i=1}^n \dfrac{1}{2}\big(\mean{A^i_0B^i_0}+\mean{A^i_1B^i_0}\big)|^{1/n}+|\displaystyle \prod_{i=1}^n \dfrac{1}{2} \big(\mean{A^i_0B^i_1}-\mean{A^i_1B^i_1}\big)|^{1/n},
}}
which leads to
\eq{\arr{
\mathcal{N}_{star}=|\displaystyle \prod_{i=1}^n \dfrac{1}{2}(\vec{a}^i_0+\vec{a}^i_1)\cdot T_{\rho_{A_iB}}\vec{b}^i_0|^{1/n}+|\displaystyle \prod_{i=1}^n \dfrac{1}{2} (\vec{a}^i_0-\vec{a}^i_1)\cdot T_{\rho_{A_iB}}\vec{b}^i_1|^{1/n}.
}}
Introducing the pairs of mutually orthogonal vectors
\eq{
\label{eq:orthogonal_vectors_changes_star}\arr{
(\vec{a}^i_0+\vec{a}^i_1)=2\; \cos \alpha_i\; \vec{n}_i \;\;\;\&\;\;\; (\vec{a}_0-\vec{a}_1)=2\; \sin \alpha_i\; \vec{n}'_i,
}}
allows us to write
\eq{\arr{
\mathcal{N}_{star}=|\displaystyle \prod_{i=1}^n \cos \alpha_i\; \vec{n}_i\cdot T_{\rho_{A_iB}}\vec{b}^i_0|^{1/n}+|\displaystyle \prod_{i=1}^n \sin \alpha_i\; \vec{n}'_i\cdot T_{\rho_{A_iB}}\vec{b}^i_1|^{1/n}.
}}
We can choose the parameters $\vec{b}^i_y$ so that they maximize the scalar products. We obtain
\eq{\arr{
\mathcal{N}^{max}_{star}=\max\bigg(|\displaystyle \prod_{i=1}^n \cos \alpha_i\; || T^{\mathbf{T}}_{\rho_{A_iB}}\vec{n}_i ||\;|^{1/n}+|\displaystyle \prod_{i=1}^n \sin \alpha_i\; || T^{\mathbf{T}}_{\rho_{A_iB}}\vec{n}'_i ||\;|^{1/n}\bigg).
}}
We can now proceed to the maximization over the parameters $\alpha_i$. Let us define the function

\eq{K(\alpha_1,...\alpha_n)= |\lambda_1 \displaystyle \prod_{i=1}^n \cos \alpha_i |^{1/n}+|\lambda_2\displaystyle \prod_{i=1}^n \sin \alpha_i |^{1/n}.
}
We can write
\eq{\dfrac{\partial K(\alpha_1,...\alpha_n)}{\partial \alpha_j}=
\dfrac{ |\lambda_2  \prod_{i=1}^n \sin \alpha_i|^{1/n}}{n} cot\alpha_j - \dfrac{|\lambda_1 \prod_{i=1}^n \cos \alpha_i|^{1/n}}{n} tan\alpha_j =0,
}
which, similarly to \eqref{eq:colored_white_imperfect_e_partial_derivatives_2values}, admits only solutions constrained by
\eq{
\tan(\alpha_j)^2=\tan(\alpha_k)^2\;\;\leftrightarrow\;\;\alpha_j = \pm \alpha_k + n\pi \;,\;n\in\mathbb{Z}\;\;\;\forall j,k.
}
This leads to
\eq{K(\alpha_1,...\alpha_n)_{max}=\max_{\alpha} \big( |\lambda_1^{1/n} \cos \alpha| +|\lambda_2^{1/n} \sin \alpha |\big)=\sqrt{\lambda_1^{2/n}+\lambda_2^{2/n}},}
which allows us to write
\eq{\arr{
\mathcal{N}^{max}_{star}=\max
\sqrt{| \displaystyle \prod_{i=1}^n || T^{\mathbf{T}}_{\rho_{A_iB}}\vec{n}_i ||\; |^{2/n}+| \displaystyle \prod_{i=1}^n || T^{\mathbf{T}}_{\rho_{A_iB}}\vec{n}'_i ||\; |^{2/n}}.
}}
Let us now define
\eq{k_1=| \displaystyle \prod_{i=2}^n || T^{\mathbf{T}}_{\rho_{A_iB}}\vec{n}_i ||\; |,\;\;\;\;\;\;
k_2=| \displaystyle \prod_{i=2}^n || T^{\mathbf{T}}_{\rho_{A_iB}}\vec{n}'_i ||\; |,}
we have that
\eq{\arr{
\mathcal{N}^{max}_{star}=\max
\sqrt{k_1^{2/n} || T^{\mathbf{T}}_{\rho_{A_iB}}\vec{n}_1 ||^{2/n}+k_2^{2/n} || T^{\mathbf{T}}_{\rho_{A_iB}}\vec{n}'_1 ||^{2/n}}.
}}

Labeling $\lambda_1,\;\lambda_2$ and $\lambda_3$ as the eigenvalues of $T_{\rho_{A_1B}}T^{\mathbf{T}}_{\rho_{A_1B}}$ (which is real and symmetric) and writing $\vec{n}_1$ and $\vec{n}'_1$ in an eigenvector basis we obtain the Lagrange multipliers problem related to the maximization of a function $f$, given the constraints $g_i$:
\eq{\arr{
f(\vec{n}_1,\;\vec{n}'_1)=\;\sqrt{\big(k_1^2 \displaystyle \sum_{i=1,2,3} \lambda_i (n^i_1)^2 \big)^{2/n}}\; + \; \sqrt{\big(k_2^2 \displaystyle \sum_{i=1,2,3} \lambda_i (n'^i_1)^2\big)^{2/n}},\\\\
g_1(\vec{n}_1)=\vec{n}_1\cdot\vec{n}_1-1,\;\;\;\;\;
g_2(\vec{n}'_1)=\vec{n}'_1\cdot\vec{n}'_1-1,\;\;\;\;\;
g_3(\vec{n}_1,\;\vec{n}'_1)=\vec{n}_1\cdot\vec{n}'_1,}}
where we considered that the values which maximize $|f(x)|$ also maximize $\sqrt{|f(x)|}$.

This Lagrangian multipliers problem can be treated similarly to \eqref{eq:lagrange_problem_bilo}, giving the same results. If $k_1>k_2$, we obtain
\eq{
f_{max}=\big(k_1 \sqrt{\lambda_1}\big)^{2/n} \; + \; \big(k_2\; \sqrt{\lambda_2}\big)^{2/n}
}
which leads to
\eq{
\mathcal{N}_{star}^{max}=\max \bigg(
\sqrt{(t^{A_1}_1)^{1/n}\;(\displaystyle \prod_{i=2}^n || T^{\mathbf{T}}_{\rho_{A_iB}}\vec{n}_i ||)^{2/n} \;+ \;(t^{A_1}_2)^{1/n}\;(  \displaystyle \prod_{i=2}^n || T^{\mathbf{T}}_{\rho_{A_iB}}\vec{n}'_i ||)^{2/n}}\bigg).
}
The proof is concluded by applying iteratively this procedure.
\end{proof}

We notice that the bilocality scenario can be seen as a particular case ($n=2$) of a star network, where $A_2\equiv C$ and $x_2\equiv z$. Moreover we emphasize that \eqref{eq:max_star_network} gives the same results that would be obtained if one performed an optimized CHSH test on a 2-qubit state were $t_1$ and $t_2$ are given by the geometric means of the parameters $t^{A_i}_1$ and $t^{A_i}_2$.

\section{Conclusions}
Generalizations of Bell's theorem to complex networks offer a new theoretical and experimental ground for further understanding quantum correlations and its practical applications in information processing. Similarly to usual Bell scenarios, understanding the set of quantum correlations we can achieve and in particular what are the optimal quantum violation of Bell inequalities is of primal importance.

In this work we have taken a step forward in this direction, deriving the optimal violation  of the bilocality inequality proposed in \cite{Branciard2010,Branciard2012} and generalized in \cite{Tavakoli2014} for the case of a star-shaped network with $n$ independent sources. Considering that the central node in the network performs arbitrary projective separable measurements and that the other parties perform projective measurements we have obtained the optimal value for the violation of the bilocality and n-locality inequalities. Our results can be understood as the generalization for complex networks of the Horodecki's criterion \cite{Horodecki1995} valid for the CHSH inequality \cite{Clauser1969}. We have analyzed in details the relation between the bilocality inequality and in particular showed that if both the quantum states cannot violate the CHSH inequality then the bilocality inequality also cannot be violated, thus precluding, in this sense, its use as a way to detect quantum correlations beyond the CHSH case. Moreover, we have showed that some quantum states can separately exhibit Bell non-local correlations, but nevertheless cannot violate the bilocality inequality when considered as a whole in the network, thus proving that not all non-local states can be used to witness non-bilocal correlations (at least according to this specific inequality).

However, all these conclusions are based on the assumption that the central node in the network performs separable measurements (that in such scenario include measurements in the Bell basis as a particular case). This immediately opens a series of interesting questions for future research. Can we achieve better violations by employing more general measurements in the central station, for instance, entangled measurements in different basis, non-maximally entangled or non-projective? Related to that, it would be highly relevant to derive new classes of network inequalities \cite{Chaves2016PRL,Rosset2016,Wolfe2016inflation}. One of the goals of generalizing Bell's theorem for complex networks is exactly the idea that since the corresponding classical models are more restrictive, it is reasonable to expect that we can find new Bell inequalities allowing us to probe the non-classical character of correlations that are local according to usual LHV models.
Can it be that separable measurements or measurement in the Bell basis allow us to detect such kind of correlations if new bilocality or n-locality inequalities are considered? And what would happen if we considered general POVM measurements in all our stations? Could we witness a whole new regime of quantum states, which at the moment, instead, admit a n-local classical description? Finally, one can wonder whether quantum states of higher dimensions (qudits) would allow for higher violations of the n-locality inequalities.\\

\emph{Note added:} During the preparation of this manuscript which contains results of a master thesis \cite{Andreoli}, we became aware of an independent work \cite{gisin2017all} preprinted in February 2017.

\begin{acknowledgments}
This work was supported by the ERC-Starting Grant 3D-QUEST (3D-Quantum Integrated Optical Simulation; grant agreement no. 307783): http://www.3dquest.eu and Brazilian ministries MEC and MCTIC. GC is supported by Becas Chile and Conicyt.
\end{acknowledgments}



%

\end{document}